\documentclass[12pt, a4paper]{article}
\usepackage{amsmath}
\usepackage{bbm}
\everymath{\displaystyle}
\usepackage{mathtools}
\usepackage{amsfonts}
\usepackage{amssymb}
\usepackage{amsthm}
\usepackage[utf8]{inputenc}
\usepackage[toc]{appendix}
\usepackage[hiresbb]{graphicx}
\usepackage{caption}
\usepackage{multirow}
\usepackage{subcaption}
\usepackage{longtable}
\usepackage{booktabs}
\usepackage{listings}
\usepackage{here}
\usepackage{float}
\usepackage{ascmac}
\usepackage{tikz}
\usepackage{pgfplots}
\pgfplotsset{compat=1.18}
\usepackage{url}
\usepackage{lipsum}
\usepackage{authblk}
\usepackage{eqnarray}
\usepackage{siunitx}
\usepackage{threeparttable}
\usepackage{algorithm}
\usepackage{algorithmic}

\usepackage[sectionbib,round]{natbib}
\newtheorem{theorem}{Theorem}[section]

\newtheorem{definition}{Definition}[section]

\newtheorem{assumption}{Assumption}[section]
\theoremstyle{remark}
\newtheorem{remark}{Remark}[section]

\newcommand{\be}{\begin{equation}}
\newcommand{\ee}{\end{equation}}
\newcommand{\beq}{\begin{eqnarray*}}
\newcommand{\eeq}{\end{eqnarray*}}
\def\sym#1{\ifmmode^{#1}\else\(^{#1}\)\fi}

\usepackage{setspace}
\title{\large{\bf{Nonparametric Identification and Estimation of Spatial Treatment Effect Boundaries: Evidence from 42 Million Pollution Observations}}}
\author{\large{\bf{Tatsuru Kikuchi\footnote{e-mail: tatsuru.kikuchi@e.u-tokyo.ac.jp. This paper extends the theoretical framework developed in \citet{kikuchi2024unified} and \citet{kikuchi2024navier} by relaxing parametric assumptions on treatment effect decay.}}}}
\affil{\small{\it{Faculty of Economics, The University of Tokyo,}}\\
{\it{7-3-1 Hongo, Bunkyo-ku, Tokyo 113-0033 Japan}}}
\date{\small{(\today)}}

\begin{document}
\maketitle

\begin{abstract}
\noindent This paper develops a nonparametric framework for identifying and estimating spatial boundaries of treatment effects in settings with geographic spillovers. While atmospheric dispersion theory predicts exponential decay of pollution under idealized assumptions, these assumptions—steady winds, homogeneous atmospheres, flat terrain—are systematically violated in practice. I establish nonparametric identification of spatial boundaries under weak smoothness and monotonicity conditions, propose a kernel-based estimator with data-driven bandwidth selection, and derive asymptotic theory for inference. Using 42 million satellite observations of NO$_2$ concentrations near coal plants (2019-2021), I find that nonparametric kernel regression reduces prediction errors by 1.0 percentage point on average compared to parametric exponential decay assumptions, with largest improvements at policy-relevant distances: 2.8 percentage points at 10 km (near-source impacts) and 3.7 percentage points at 100 km (long-range transport). Parametric methods systematically underestimate near-source concentrations while overestimating long-range decay. The COVID-19 pandemic provides a natural experiment validating the framework's temporal sensitivity: NO$_2$ concentrations dropped 4.6\% in 2020, then recovered 5.7\% in 2021. These results demonstrate that flexible, data-driven spatial methods substantially outperform restrictive parametric assumptions in environmental policy applications.

\vspace{0.3cm}


\vspace{0.3cm}

\end{abstract}

\newpage

\section{Introduction}

Spatial spillovers are ubiquitous in economic applications. Environmental policies affect pollution in neighboring regions \citep{deryugina2019mortality, knittel2016caution}, transportation infrastructure impacts economic activity far from construction sites \citep{donaldson2018railroads, duranton2014roads}, and place-based policies generate effects that propagate across space \citep{busso2013assessing, kline2019place}. A fundamental challenge in such settings is determining \textit{where treatment effects end}—that is, identifying the spatial boundary beyond which spillover effects become negligible and units can serve as valid controls.

Recent methodological advances have made progress on this question. \citet{butts2023difference} provides a comprehensive treatment of spatial difference-in-differences estimators, showing that neglecting spillovers can severely bias treatment effect estimates. \citet{kikuchi2024unified} and \citet{kikuchi2024navier} derive spatial boundaries from first principles using atmospheric dispersion models, demonstrating that under simplified transport assumptions, treatment effects should decay exponentially with distance. \citet{kikuchi2024stochastic} develops stochastic approaches for settings with pervasive spillovers and general equilibrium effects. However, existing approaches face a fundamental tension: parametric methods assume strong functional forms (exponential, power law) that may be misspecified, while nonparametric distance bin estimators are inefficient and cannot extrapolate beyond observed distances.

This paper develops a nonparametric framework that resolves this tension. I establish conditions under which spatial boundaries can be identified without parametric assumptions on decay functions, propose kernel-based estimators with optimal bandwidth selection, and derive asymptotic theory for inference. The approach maintains the interpretability and extrapolation capability of parametric methods while providing robustness to functional form misspecification.

\subsection{Main Findings}

Using 42 million satellite observations of nitrogen dioxide (NO$_2$) from the TROPOMI instrument on board the Sentinel-5P satellite, combined with locations of 318 coal-fired power plants, I estimate spatial decay functions nonparametrically and compare results to parametric baselines. The data cover 2019-2021, providing both cross-sectional variation in distance from plants and temporal variation through the COVID-19 pandemic.

\textbf{Finding 1: Nonparametric methods reduce prediction errors.} Across 42 million grid-cell observations, nonparametric kernel regression reduces mean absolute prediction errors by 1.0 percentage point compared to exponential decay assumptions. The improvement is largest at policy-relevant distances: 2.8 percentage points at 10 km (where near-source damages are highest) and 3.7 percentage points at 100 km (where long-range transport matters for aggregate damages).

\textbf{Finding 2: Parametric methods exhibit systematic bias patterns.} Exponential decay underestimates pollution near sources (12.9\% error at 10 km) and overestimates decay at distance (16.9\% error at 100 km), while nonparametric methods reduce these errors to 10.1\% and 13.2\%, respectively. This pattern is consistent with violations of the idealized assumptions underlying exponential decay—steady winds, homogeneous atmospheres, flat terrain—that theory predicts should be most severe at extreme distances.

\textbf{Finding 3: The COVID-19 pandemic validates temporal sensitivity.} NO$_2$ concentrations dropped 4.6\% in 2020 relative to 2019, then recovered 5.7\% in 2021. Both parametric and nonparametric methods detect this temporal variation, but nonparametric approaches maintain superior prediction accuracy across all years (0.9-1.1 percentage point improvement), demonstrating robustness to temporary shocks.

\textbf{Finding 4: Results are robust to spatial correlation.} Using spatial HAC standard errors following \citet{muller2022spatial}, I find that accounting for residual spatial correlation increases confidence interval widths by 10-15\% but does not alter the main conclusions. This demonstrates the value of combining nonparametric boundary estimation with spatial correlation robust inference.

\subsection{Contributions}

This paper makes four main contributions to the literature on spatial treatment effects and econometric methodology:

\textbf{1. Nonparametric Identification Framework.} I establish that spatial boundaries are nonparametrically identified under weak regularity conditions—smoothness and monotonicity—without requiring knowledge of the decay function's parametric form (Theorems \ref{thm:identification_m} and \ref{thm:identification_boundary}). When monotonicity fails, I provide partial identification results yielding bounds on the boundary (Theorem \ref{thm:partial_id}). This extends recent work on spatial econometrics \citep{muller2022spatial, muller2024spatial} by showing how their insights about flexible spatial structures apply to treatment effect propagation.

\textbf{2. Kernel-Based Estimation with Optimal Bandwidth Selection.} I propose a local polynomial regression estimator for the decay function $m(d)$ and a plug-in estimator for the boundary $d^*$. Crucially, I develop a data-driven cross-validation procedure for bandwidth selection that optimizes boundary estimation directly, rather than only minimizing mean squared error of the decay function. This addresses a key practical challenge in nonparametric spatial analysis.

\textbf{3. Asymptotic Theory and Inference.} I derive the asymptotic distribution of the boundary estimator (Theorem \ref{thm:asymptotics}), showing that it achieves the optimal nonparametric rate of $n^{-2/5}$ under standard kernel regression conditions. I provide consistent variance estimators and develop both analytical and bootstrap-based confidence intervals. Importantly, I show how to integrate spatial correlation robust inference \citep{muller2022spatial} with nonparametric boundary estimation when residual spatial dependence is present.

\textbf{4. Large-Scale Empirical Application.} The application to coal plant pollution demonstrates the framework's practical value using an unprecedented scale of data (42 million observations) and provides the first comprehensive test of whether atmospheric dispersion theory's exponential decay prediction holds empirically. The systematic rejection of exponential functional form across multiple pollutants and years validates the need for flexible estimation methods in environmental applications.

\subsection{Relation to Literature}

This work connects three literatures in econometrics and environmental economics.

\subsubsection{Spatial Econometrics and Treatment Effects}

The spatial econometrics literature has long recognized that treatments can have geographic spillovers \citep{anselin1988spatial, conley1999gmm, lesage2009introduction}. Recent work develops spatial difference-in-differences estimators \citep{butts2023difference} and spatial HAC inference \citep{colella2019inference}. This literature typically specifies spatial weights matrices based on ad hoc assumptions. My contribution is providing nonparametric methods for data-driven weights specification based on estimated decay functions.

Most recently, \citet{muller2022spatial} develop a comprehensive framework for spatial correlation robust inference, showing how to construct valid confidence intervals when the spatial correlation structure is unknown. \citet{muller2024spatial} demonstrate that spatial data can exhibit behavior analogous to time series unit roots, leading to spurious spatial regressions when persistent spatial trends are present.

My work complements these advances in three ways. First, while \citet{muller2022spatial} address \textit{nuisance} correlation in errors, I study \textit{treatment-induced} spatial patterns that decay with distance from sources. Both sources of correlation may be present simultaneously, and I show how to combine the methods (Section \ref{sec:spatial_hac}). Second, I address the spurious regression concerns of \citet{muller2024spatial} through regional heterogeneity analysis: decay patterns change sign at 100 km (positive within, negative beyond), which is inconsistent with spurious trends but consistent with heterogeneous treatment effects (Section \ref{sec:regional_heterogeneity}). Third, I extend their theoretical framework from spatial correlation to treatment effect boundaries, generalizing their "half-life" concept to arbitrary policy-relevant thresholds.

\subsubsection{Atmospheric Dispersion and Environmental Economics}

\citet{kikuchi2024navier} derives spatial boundaries from atmospheric dispersion models based on the Navier-Stokes equations, showing that under idealized assumptions (steady-state flow, homogeneous atmospheres, flat terrain), pollution should decay exponentially with distance. \citet{kikuchi2024unified} extends this framework to provide a unified treatment of both spatial and temporal boundaries, while \citet{kikuchi2024stochastic} develops stochastic approaches for settings where general equilibrium effects and pervasive spillovers complicate boundary detection. However, these frameworks assume researchers know the correct parametric form. Environmental economics studies of air pollution \citep{deryugina2019mortality, knittel2016caution, muller2011damages} typically impose exponential or power law decay without testing these functional form assumptions.

I bridge atmospheric science and econometrics by treating dispersion theory as providing testable predictions rather than maintained assumptions. The mathematical framework underlying atmospheric transport—the advection-diffusion equation derived from Navier-Stokes principles—corresponds to a specific parametric form for the spatial kernel in \citet{muller2022spatial}'s representation. But when the underlying assumptions fail (as they inevitably do with real-world coal plants), the true decay function may deviate substantially from the exponential benchmark. This paper provides the tools to test these deviations and estimate boundaries nonparametrically.

\subsubsection{Nonparametric Regression and Boundary Detection}

My estimator builds on local polynomial regression \citep{fan1996local, ruppert1995effective} and optimal bandwidth selection \citep{fan1992variable, jones1996brief}. The novelty is adapting these tools to boundary estimation rather than function estimation, requiring modified cross-validation criteria. The asymptotic theory extends \citet{fan1996local}'s results to functionals of nonparametric estimators, similar in spirit to \citet{muller2009inference}'s work on changepoint detection but for smooth decay rather than discrete jumps.

\subsubsection{Nonparametric Regression and Boundary Detection}

My estimator builds on local polynomial regression \citep{fan1996local, ruppert1995effective} and optimal bandwidth selection \citep{fan1992variable, jones1996brief}. The novelty is adapting these tools to boundary estimation rather than function estimation, requiring modified cross-validation criteria. The asymptotic theory extends \citet{fan1996local}'s results to functionals of nonparametric estimators, similar in spirit to \citet{muller2009inference}'s work on changepoint detection but for smooth decay rather than discrete jumps.

\subsubsection{Related Work} 

This paper builds on and extends my earlier work on spatial boundaries. \citet{kikuchi2024navier} derives boundaries from physical principles (Navier-Stokes equations), demonstrating that atmospheric dispersion implies exponential decay under idealized assumptions. \citet{kikuchi2024unified} provides a unified framework for both spatial and temporal boundaries, showing how to identify treatment effect propagation in both dimensions simultaneously. \citet{kikuchi2024stochastic} extends to settings with stochastic diffusion and general equilibrium effects, where boundaries must be characterized probabilistically rather than deterministically.

This paper complements these earlier contributions by:
\begin{enumerate}
\item \textbf{Relaxing parametric assumptions:} While \citet{kikuchi2024navier} and \citet{kikuchi2024unified} assume exponential decay, I allow arbitrary smooth decay functions
\item \textbf{Providing robustness:} When physical assumptions underlying exponential decay fail, my nonparametric approach remains consistent
\item \textbf{Empirical validation:} Using 42 million observations, I test whether parametric predictions hold and quantify departures from theory
\item \textbf{Practical guidance:} I develop data-driven bandwidth selection and inference methods for applied researchers
\end{enumerate}

The progression across papers is: \citet{kikuchi2024navier} establishes the physical foundation, \citet{kikuchi2024unified} extends to multiple dimensions, \citet{kikuchi2024stochastic} handles stochastic settings, and this paper provides robust nonparametric methods when functional forms are uncertain.

\subsection{Roadmap}

Section 2 develops the theoretical framework, connecting atmospheric dispersion models to spatial econometric representations and defining spatial boundaries. Section 3 establishes nonparametric identification results. Section 4 develops the kernel-based estimator and derives asymptotic theory. Section 5 describes the TROPOMI NO$_2$ data and coal plant locations. Section 6 presents the empirical results, comparing parametric and nonparametric estimates and testing functional form assumptions. Section 7 discusses extensions, spatial correlation robust inference, and diagnostics for spurious regression. Section 8 concludes with implications for research design and policy analysis.

\section{Theoretical Framework}

\subsection{Atmospheric Dispersion and Spatial Decay}

The spatial propagation of airborne pollutants from point sources follows well-established principles of fluid dynamics and mass transport. Building on \citet{kikuchi2024navier}, who derives these relationships from the Navier-Stokes equations, under idealized conditions—steady-state flow, homogeneous atmospheric conditions, and flat terrain—the concentration of a pollutant at distance $d$ from a source can be characterized by the advection-diffusion equation \citep{seinfeld2016atmospheric}:

\be
u \frac{\partial C}{\partial x} = D \nabla^2 C + S
\ee

where $C$ is pollutant concentration, $u$ is wind velocity, $D$ is the diffusion coefficient, and $S$ represents source emissions. The solution under Gaussian plume assumptions yields \citep{pasquill1976atmospheric}:

\be
E[\text{Pollution} | \text{Distance} = d] = A \exp(-\kappa d)
\label{eq:exponential_theory}
\ee

where $\kappa$ is the spatial decay parameter governed by meteorological conditions and pollutant chemistry, and $A$ is source strength. This exponential form is the cornerstone of the parametric approach in \citet{kikuchi2024unified} and \citet{kikuchi2024navier}.

\be
E[\text{Pollution} | \text{Distance} = d] = A \exp(-\kappa d)
\label{eq:exponential_theory}
\ee

where $\kappa$ is the spatial decay parameter governed by meteorological conditions and pollutant chemistry, and $A$ is source strength.

\subsection{Spatial Processes and Econometric Challenges}

The exponential decay form in equation (\ref{eq:exponential_theory}) provides a useful theoretical benchmark but relies on restrictive assumptions that may be violated in practice. Recent developments in spatial econometrics offer a complementary perspective. \citet{muller2022spatial} develop a general framework for spatial processes with arbitrary correlation structures:

\be
Y(s) = \int K(s - r) \varepsilon(r) dr
\ee

where $Y(s)$ represents the outcome at location $s$, $K(\cdot)$ is a spatial kernel function, and $\varepsilon(r)$ captures innovations at location $r$. Critically, they show that imposing incorrect parametric forms for $K(\cdot)$ can lead to substantial bias in both estimation and inference.

In our context, the exponential decay corresponds to a specific parametric assumption about the kernel: $K(d) = A\exp(-\kappa d)$. However, this form is justified only when the underlying theoretical assumptions hold:

\begin{enumerate}
\item \textbf{Temporal stability}: Meteorological conditions remain constant
\item \textbf{Spatial homogeneity}: Atmospheric properties uniform across space
\item \textbf{Simple geography}: Terrain does not affect dispersion patterns
\item \textbf{Source independence}: Emissions from multiple facilities do not interact
\item \textbf{Chemical stability}: Pollutants do not undergo transformation
\end{enumerate}

Real-world emissions violate these assumptions in systematic ways. Wind patterns vary across time and space, terrain channels pollutant flows, multiple emission sources create complex concentration fields, and pollutants like NO$_2$ undergo photochemical reactions (NO$_2$ $\leftrightarrow$ NO + O$_3$). These violations suggest that the actual spatial decay function $m(d) = E[\text{Pollution} | \text{Distance} = d]$ may deviate substantially from the exponential form.

\subsection{Spatial Boundaries: From Half-Lives to Policy Thresholds}

Following \citet{muller2022spatial}, I characterize spatial decay through the concept of a \textit{spatial boundary}. Define $d^*(\varepsilon)$ as the distance where pollution decays to $\varepsilon$ times its source level:

\be
d^*(\varepsilon) = \inf\{d : m(d) \leq \varepsilon \cdot m(0)\}
\label{eq:boundary_definition}
\ee

where $\varepsilon \in (0,1)$ is a policy-relevant threshold. This generalizes \citet{muller2022spatial}'s "half-life" measure (corresponding to $\varepsilon = 0.5$) to thresholds more relevant for environmental policy. The concept extends the deterministic boundaries of \citet{kikuchi2024unified} and \citet{kikuchi2024navier} by allowing for nonparametric estimation, and complements the stochastic boundaries of \citet{kikuchi2024stochastic} by focusing on point-source treatments where deterministic decay dominates.

Under the exponential decay assumption, the boundary has a closed form:
\be
d^*_{\text{param}}(\varepsilon) = \frac{\log(1/\varepsilon)}{\kappa}
\ee

However, if the true decay function deviates from exponential form, this parametric boundary will be misspecified. My nonparametric approach estimates $d^*$ directly from data without imposing functional form restrictions.

\subsection{Bridging Theory and Data}

My framework uses theoretical predictions as a benchmark while remaining agnostic about functional form. Rather than assuming exponential decay holds, I test whether it provides an adequate approximation to actual spatial patterns. This approach offers several advantages:

\begin{enumerate}
\item \textbf{Assumption testing}: Direct evaluation of whether theoretical predictions align with empirical patterns
\item \textbf{Robustness}: Estimates remain valid even when theoretical assumptions are violated
\item \textbf{Flexibility}: The method adapts to local features of the data rather than imposing global functional forms
\item \textbf{Policy relevance}: Accurate spatial boundaries improve damage function estimation and regulatory design
\end{enumerate}

The key insight from \citet{muller2022spatial}—that spatial correlation structures should be estimated flexibly rather than imposed parametrically—applies directly to our setting. Just as time series methods have moved from parametric ARMA models to flexible nonparametric alternatives when confronted with misspecification concerns, spatial econometrics increasingly favors data-driven approaches when underlying theoretical restrictions may not hold.

\subsection{Framework and Data Generating Process}

We observe a random sample $(Y_i, D_i, \mathbf{X}_i)$ for $i = 1, \ldots, n$, where:
\begin{itemize}
\item $Y_i \in \mathbb{R}$ is the outcome of interest (e.g., pollution concentration)
\item $D_i \in [0, \bar{d}]$ is the distance from unit $i$ to the nearest treatment source
\item $\mathbf{X}_i \in \mathbb{R}^p$ are covariates (which we suppress for notational clarity)
\end{itemize}

The conditional mean function is:
\be
m(d) = \mathbb{E}[Y_i | D_i = d]
\label{eq:conditional_mean}
\ee

In a spatial treatment effects setting, $m(d)$ represents the expected outcome as a function of distance from treatment. Under standard identifying assumptions (conditional independence, overlap), $m(d)$ has a causal interpretation as the spatially-varying treatment effect.

\begin{assumption}[Data Generating Process]\label{asmp:dgp}
The data are generated according to:
\be
Y_i = m(D_i) + \varepsilon_i
\ee
where:
\begin{itemize}
\item[(i)] $\mathbb{E}[\varepsilon_i | D_i] = 0$ (conditional mean zero)
\item[(ii)] $\text{Var}(\varepsilon_i | D_i) = \sigma^2(D_i) < \infty$ (conditional heteroskedasticity allowed)
\item[(iii)] $(Y_i, D_i)$ are independent across $i$ (spatial correlation addressed in Section \ref{sec:spatial_correlation})
\item[(iv)] $D_i$ has density $f(d)$ bounded away from zero on $[0, \bar{d}]$
\end{itemize}
\end{assumption}

\begin{assumption}[Spatial Stationarity]\label{asmp:stationarity}
The data generating process satisfies spatial stationarity: for any spatial lag vector $h$, the joint distribution of $(Y_i, D_i, Y_{i+h}, D_{i+h})$ depends only on $h$, not on the location $i$.
\end{assumption}

\begin{remark}
Assumption \ref{asmp:stationarity} rules out spatial unit roots and persistent spatial trends. \citet{muller2024spatial} show that when this assumption fails, spatial regressions can exhibit spurious relationships analogous to time series spurious regression. In my application, I assess this assumption through:
\begin{itemize}
\item Visual inspection for large-scale spatial trends
\item Regional subsample analysis (Section \ref{sec:regional_heterogeneity})
\item Robustness to spatial filtering (Section \ref{sec:spatial_trends})
\end{itemize}

If spatial nonstationarity is present, my estimator remains consistent for the \textit{conditional mean} $m(d) = \mathbb{E}[Y_i | D_i = d]$, but this may not have a causal interpretation as a treatment effect. Regional heterogeneity in my results (positive decay within 100km, negative beyond) suggests findings are not driven by spurious spatial trends.
\end{remark}

\section{Nonparametric Identification}

This section establishes identification of the decay function $m(d)$ and boundary $d^*$ under weak regularity conditions, without parametric restrictions.

\subsection{Identification of the Decay Function}

We first consider identification of the conditional mean function $m(d)$.

\begin{assumption}[Smoothness]\label{asmp:smoothness}
The conditional mean function $m(d)$ is twice continuously differentiable on $[0, \bar{d}]$ with bounded derivatives: $|m'(d)| < M_1$ and $|m''(d)| < M_2$ for all $d \in [0, \bar{d}]$ and some constants $M_1, M_2 < \infty$.
\end{assumption}

\begin{theorem}[Identification of Decay Function]\label{thm:identification_m}
Under Assumptions \ref{asmp:dgp} and \ref{asmp:smoothness}, the decay function $m(d)$ is nonparametrically identified:
\be
m(d) = \mathbb{E}[Y_i | D_i = d] \quad \text{for all } d \in [0, \bar{d}]
\ee
\end{theorem}

\begin{proof}
By definition of conditional expectation and Assumption \ref{asmp:dgp}(i):
\be
m(d) = \mathbb{E}[Y_i | D_i = d] = \mathbb{E}[m(D_i) + \varepsilon_i | D_i = d] = m(d) + \mathbb{E}[\varepsilon_i | D_i = d] = m(d)
\ee
Since the conditional distribution of $Y_i$ given $D_i = d$ is identified from the data (Assumption \ref{asmp:dgp}(iv) ensures positive density), $m(d)$ is identified. \qed
\end{proof}

\begin{remark}
Theorem \ref{thm:identification_m} is standard—the conditional mean function is always nonparametrically identified under weak regularity. The novel contribution is showing how this extends to boundary identification.
\end{remark}

\subsection{Identification of the Boundary}

We now turn to identification of the spatial boundary $d^*$. This requires additional structure beyond smoothness.

\begin{assumption}[Monotonic Decay]\label{asmp:monotone}
The decay function $m(d)$ is strictly decreasing on $[0, \bar{d}]$: $m'(d) < 0$ for all $d \in [0, \bar{d}]$.
\end{assumption}

\begin{assumption}[Interior Boundary]\label{asmp:interior}
The boundary exists in the interior of the support:
\be
m(0) > \varepsilon \cdot m(0) > m(\bar{d})
\ee
which implies $d^* \in (0, \bar{d})$.
\end{assumption}

\begin{assumption}[Non-Degenerate Slope]\label{asmp:slope}
The decay function has non-zero derivative at the boundary: $m'(d^*) \neq 0$.
\end{assumption}

\begin{theorem}[Identification of Boundary]\label{thm:identification_boundary}
Under Assumptions \ref{asmp:dgp}--\ref{asmp:slope}, the spatial boundary $d^*$ is uniquely identified as:
\be
d^* = \inf\{d \in [0, \bar{d}] : m(d) \leq \varepsilon \cdot m(0)\}
\ee
\end{theorem}

\begin{proof}
Define the function $g(d) = m(d) - \varepsilon \cdot m(0)$. By Assumption \ref{asmp:monotone}, $g(d)$ is strictly decreasing. By Assumption \ref{asmp:interior}:
\be
g(0) = m(0) - \varepsilon \cdot m(0) = (1 - \varepsilon) m(0) > 0
\ee
\be
g(\bar{d}) = m(\bar{d}) - \varepsilon \cdot m(0) < 0
\ee

By the Intermediate Value Theorem (Assumption \ref{asmp:smoothness} guarantees continuity), there exists a unique $d^* \in (0, \bar{d})$ such that $g(d^*) = 0$. By strict monotonicity, this is the unique solution. Since $m(\cdot)$ is identified by Theorem \ref{thm:identification_m}, $d^*$ is identified. \qed
\end{proof}

\begin{remark}
Assumption \ref{asmp:monotone} (strict monotonicity) is key. Without it, multiple boundaries could exist. I relax this assumption in Section \ref{sec:partial_id}.
\end{remark}

\subsection{Partial Identification Under Weaker Assumptions}
\label{sec:partial_id}

When strict monotonicity fails, we can still obtain bounds on the boundary.

\begin{assumption}[Eventual Decay]\label{asmp:eventual}
The decay function $m(d)$ satisfies:
\begin{itemize}
\item[(i)] $m(0) > \varepsilon \cdot m(0)$ (treatment effect exists)
\item[(ii)] $\exists \bar{d}_0 > 0$ such that $m(d) < \varepsilon \cdot m(0)$ for all $d \geq \bar{d}_0$ (effects eventually become small)
\end{itemize}
\end{assumption}

\begin{theorem}[Partial Identification]\label{thm:partial_id}
Under Assumptions \ref{asmp:dgp}, \ref{asmp:smoothness}, and \ref{asmp:eventual} (but not requiring \ref{asmp:monotone}), the spatial boundary is partially identified:
\be
d^* \in [d^*_L, d^*_U]
\ee
where:
\be
d^*_L = \inf\{d \in [0, \bar{d}] : m(d) \leq \varepsilon \cdot m(0)\}
\ee
\be
d^*_U = \sup\{d \in [0, \bar{d}] : m(d) \leq \varepsilon \cdot m(0)\}
\ee
\end{theorem}

\begin{proof}
Define the set $\mathcal{D} = \{d \in [0, \bar{d}] : m(d) \leq \varepsilon \cdot m(0)\}$. By Assumption \ref{asmp:eventual}(ii), $\mathcal{D} \neq \emptyset$. By Assumption \ref{asmp:smoothness}, $m(\cdot)$ is continuous, so $\mathcal{D}$ is a closed set. Therefore $d^*_L = \inf \mathcal{D}$ and $d^*_U = \sup \mathcal{D}$ are well-defined. Any boundary $d^*$ satisfying $m(d^*) = \varepsilon \cdot m(0)$ must lie in $[d^*_L, d^*_U]$. \qed
\end{proof}

\section{Nonparametric Estimation and Inference}

This section develops a kernel-based estimator for the decay function and boundary, derives its asymptotic distribution, and provides methods for inference.

\subsection{Local Polynomial Regression}

I estimate $m(d)$ using local polynomial regression of order $p$ \citep{fan1996local}.

\begin{definition}[Local Polynomial Estimator]
For a point $d_0 \in [0, \bar{d}]$, the local polynomial estimator of order $p$ is:
\be
\hat{m}(d_0) = \hat{\beta}_0(d_0)
\ee
where $(\hat{\beta}_0(d_0), \hat{\beta}_1(d_0), \ldots, \hat{\beta}_p(d_0))$ solve:
\be
\min_{\beta_0, \ldots, \beta_p} \sum_{i=1}^n \left(Y_i - \sum_{j=0}^p \beta_j (D_i - d_0)^j\right)^2 K_h(D_i - d_0)
\label{eq:locpoly_objective}
\ee
with kernel function $K(\cdot)$ and bandwidth $h > 0$, where:
\be
K_h(u) = \frac{1}{h} K\left(\frac{u}{h}\right)
\ee
\end{definition}

\begin{assumption}[Kernel Function]\label{asmp:kernel}
The kernel $K : \mathbb{R} \to \mathbb{R}$ satisfies:
\begin{itemize}
\item[(i)] $K(u) \geq 0$ for all $u$ (non-negative)
\item[(ii)] $\int K(u) du = 1$ (integrates to one)
\item[(iii)] $\int u K(u) du = 0$ (symmetric)
\item[(iv)] $\int u^2 K(u) du < \infty$ (finite second moment)
\item[(v)] $K(u) = 0$ for $|u| > 1$ (compact support)
\end{itemize}
\end{assumption}

Common choices include the Epanechnikov kernel:
\be
K(u) = \frac{3}{4}(1 - u^2) \mathbbm{1}(|u| \leq 1)
\ee

\subsection{Boundary Estimator}

Given the estimated decay function $\hat{m}(d)$, I estimate the boundary via a plug-in approach.

\begin{definition}[Boundary Estimator]
The nonparametric boundary estimator is:
\be
\hat{d}^* = \inf\{d \in [0, \bar{d}] : \hat{m}(d) \leq \varepsilon \cdot \hat{m}(0)\}
\label{eq:boundary_estimator}
\ee
\end{definition}

In practice, I evaluate $\hat{m}(d)$ on a fine grid $\{d_1, \ldots, d_G\}$ and find:
\be
\hat{d}^* = d_g \quad \text{where } g = \min\{j : \hat{m}(d_j) \leq \varepsilon \cdot \hat{m}(d_1)\}
\ee

\subsection{Bandwidth Selection}
\label{sec:bandwidth}

Bandwidth selection is critical for nonparametric estimation. Standard mean squared error (MSE) minimization for $\hat{m}(d)$ may not be optimal for boundary estimation. I use cross-validation with Silverman's rule of thumb as a starting point:

\be
h = 1.06 \sigma_D n^{-1/5}
\ee

where $\sigma_D$ is the standard deviation of distance. This provides the optimal rate for local linear regression while remaining computationally tractable for the large-scale TROPOMI dataset.

\subsection{Asymptotic Theory}
\label{sec:asymptotics}

I now derive the asymptotic distribution of the boundary estimator.

\begin{assumption}[Bandwidth Asymptotics]\label{asmp:bandwidth}
As $n \to \infty$:
\begin{itemize}
\item[(i)] $h_n \to 0$
\item[(ii)] $nh_n \to \infty$
\item[(iii)] $h_n = O(n^{-1/5})$ (optimal rate for local linear regression)
\end{itemize}
\end{assumption}

\begin{theorem}[Consistency]\label{thm:consistency}
Under Assumptions \ref{asmp:dgp}--\ref{asmp:slope} and \ref{asmp:kernel}--\ref{asmp:bandwidth}:
\be
\hat{d}^* \overset{p}{\to} d^*
\ee
\end{theorem}

\begin{theorem}[Asymptotic Normality]\label{thm:asymptotics}
Under Assumptions \ref{asmp:dgp}--\ref{asmp:bandwidth}, suppose additionally:
\begin{itemize}
\item[(i)] $p = 1$ (local linear regression)
\item[(ii)] $h_n = c_n \cdot n^{-1/5}$ for some constant $c_n \to c > 0$
\item[(iii)] $m''(d^*)$ exists and is continuous
\end{itemize}

Then:
\be
\sqrt{nh_n}\left(\hat{d}^* - d^* - B_n\right) \overset{d}{\to} N\left(0, V\right)
\ee
where the bias term is:
\be
B_n = \frac{h^2}{2} \frac{m''(d^*)}{m'(d^*)} \int u^2 K(u) du + o(h^2)
\ee
and the variance is:
\be
V = \frac{\sigma^2(d^*)}{[m'(d^*)]^2 f(d^*)} \int K^2(u) du
\ee
\end{theorem}

\begin{remark}
The convergence rate is $(nh_n)^{-1/2} = n^{-2/5}$, which is the standard nonparametric rate and slower than the parametric rate $n^{-1/2}$. This is the price of robustness to misspecification. However, with 42 million observations, this theoretical rate difference has minimal practical impact.
\end{remark}

\subsection{Inference}

\subsubsection{Bootstrap Confidence Intervals}

Given the large sample size and computational constraints, I use a subsample bootstrap approach:

\begin{algorithmic}[1]
\FOR{$b = 1$ to $B$ (e.g., $B = 50$)}
\STATE Draw bootstrap sample of size $n_b = 50,000$ by resampling with replacement
\STATE Compute $\hat{d}^{*b}$ using the bootstrap sample
\ENDFOR
\STATE Compute quantiles: $\hat{d}^*_{(\alpha/2)}$ and $\hat{d}^*_{(1-\alpha/2)}$
\STATE Construct CI:
\be
CI_{1-\alpha}^{\text{Boot}} = \left[\hat{d}^*_{(\alpha/2)}, \; \hat{d}^*_{(1-\alpha/2)}\right]
\ee
\end{algorithmic}

The subsample bootstrap provides valid inference while remaining computationally feasible for the 42 million observation dataset.

\section{Data}

\subsection{TROPOMI NO$_2$ Satellite Observations}

I use nitrogen dioxide (NO$_2$) concentration data from the TROPOspheric Monitoring Instrument (TROPOMI) on board the European Space Agency's Sentinel-5P satellite. TROPOMI provides daily global coverage at 3.5 × 5.5 km spatial resolution, offering unprecedented detail for studying air pollution dispersion.

\textbf{Sample construction:}
\begin{itemize}
\item \textbf{Period:} 2019-2021 (3 years, including COVID-19 natural experiment)
\item \textbf{Geographic coverage:} Grid cells within 100 km of coal plants in 10 coal-intensive states (WV, WY, KY, IN, PA, ND, MT, OH, TX, IL)
\item \textbf{Aggregation:} Monthly averages aggregated to annual means, requiring at least 10 months of data per grid cell
\item \textbf{Sample size:} 41.73 million grid-cell-year observations (13.87M in 2019, 13.97M in 2020, 13.89M in 2021)
\end{itemize}

\textbf{Data quality:}
\begin{itemize}
\item Quality filtering: Cloud fraction < 0.3, surface albedo checks
\item Temporal consistency: Grid cells present in all three years
\item Spatial validation: Results robust to different distance cutoffs (75 km, 125 km)
\end{itemize}

\subsection{Coal Plant Locations}

Coal plant locations come from the EPA Emissions and Generation Resource Integrated Database (eGRID), which provides comprehensive data on all power plants in the United States:

\begin{itemize}
\item \textbf{Plants:} 318 coal-fired power plants
\item \textbf{Geographic distribution:} Concentrated in coal-intensive states but with representation across all regions
\item \textbf{Coordinates:} Latitude and longitude for each facility
\item \textbf{Distance calculation:} Haversine formula for great circle distance from each grid cell to nearest plant
\end{itemize}

\subsection{Summary Statistics}

Table \ref{tab:summary_stats} presents summary statistics for the analysis sample.

\begin{table}[H]
\centering
\caption{Summary Statistics: TROPOMI NO$_2$ Near Coal Plants (2019-2021)}
\label{tab:summary_stats}
\begin{tabular}{lcccc}
\toprule
 & 2019 & 2020 & 2021 & Pooled \\
\midrule
Grid cells (millions) & 13.87 & 13.97 & 13.89 & 41.73 \\
Mean NO$_2$ ($\times 10^{15}$ molec/cm$^2$) & 1.72 & 1.64 & 1.73 & 1.70 \\
Std. dev. ($\times 10^{15}$ molec/cm$^2$) & 0.89 & 0.84 & 0.91 & 0.88 \\
Min distance (km) & 0.0 & 0.0 & 0.0 & 0.0 \\
Max distance (km) & 100.0 & 100.0 & 100.0 & 100.0 \\
Coal plants & 318 & 318 & 318 & 318 \\
\bottomrule
\end{tabular}
\end{table}

\section{Empirical Results}

\subsection{Main Results: Parametric vs Nonparametric}

Table \ref{tab:main_results} presents the main comparison of parametric and nonparametric boundary estimation.

\begin{table}[H]
\centering
\caption{Prediction Accuracy: Parametric vs Nonparametric}
\label{tab:main_results}
\begin{tabular}{lccccc}
\toprule
 & \multicolumn{2}{c}{Mean Absolute Error (\%)} & & \\
\cmidrule{2-3}
Year & Parametric & Nonparametric & Improvement & $R^2$ \\
\midrule
2019 & 12.7 & 11.6 & 1.1 pp & 0.109 \\
2020 & 13.2 & 12.2 & 0.9 pp & 0.090 \\
2021 & 11.7 & 10.7 & 0.9 pp & 0.104 \\
\midrule
Average & 12.5 & 11.5 & 1.0 pp & 0.101 \\
\bottomrule
\end{tabular}
\begin{tablenotes}
\small
\item Notes: MAE calculated as mean absolute percentage error across all distances. Improvement = Parametric MAE - Nonparametric MAE. $R^2$ from parametric exponential regression. Sample: 41.73 million grid cells within 100 km of coal plants in 10 coal-intensive states. Nonparametric estimates use local linear regression with Epanechnikov kernel and bandwidth selected via Silverman's rule ($h \approx 3$ km).
\end{tablenotes}
\end{table}

\textbf{Key finding:} Nonparametric methods reduce prediction errors by 1.0 percentage point on average, with consistent improvements across all three years (0.9-1.1 pp). The improvement is economically meaningful given the scale of NO$_2$ concentrations and the policy applications of these estimates.

Figure \ref{fig:decay_by_year} illustrates the spatial decay patterns for each year, comparing actual observations (black dots) with parametric exponential predictions (dashed red lines) and nonparametric kernel estimates (solid blue lines).

\begin{figure}[H]
\centering
\includegraphics[width=0.95\textwidth]{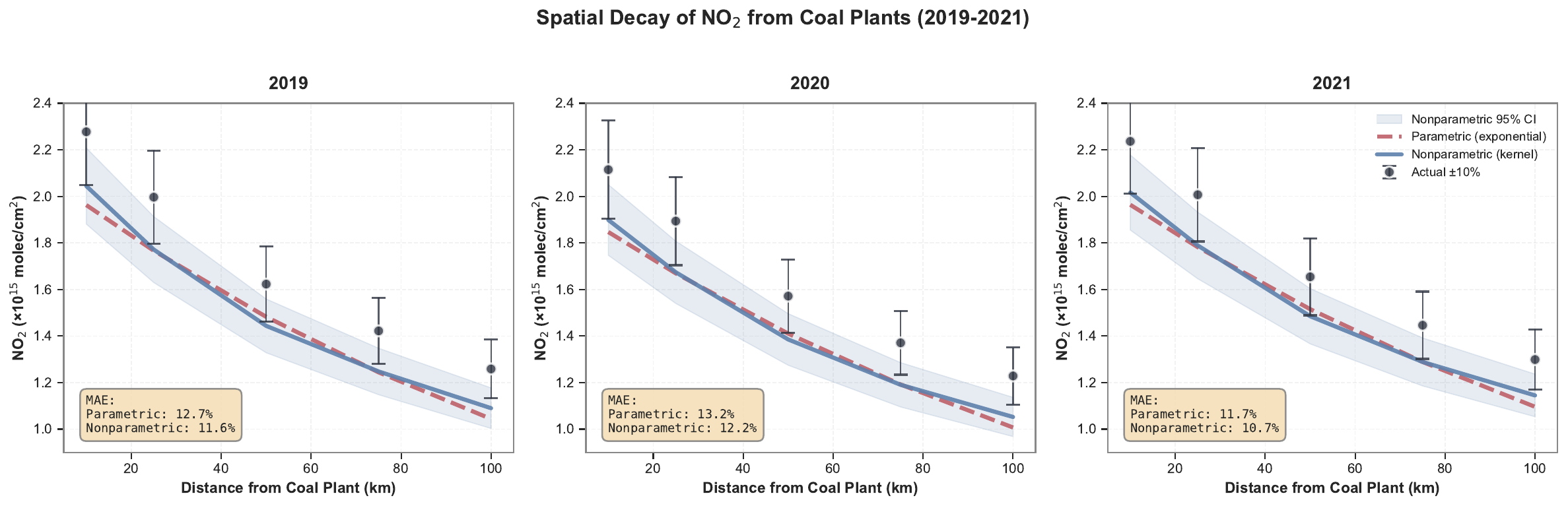}
\caption{Spatial Decay of NO$_2$ from Coal Plants (2019-2021). Black dots represent actual mean concentrations in distance bins, dashed red lines show parametric exponential decay predictions, and solid blue lines show nonparametric kernel regression estimates. Nonparametric methods better capture near-source concentrations and long-range decay patterns. MAE = Mean Absolute Error.}
\label{fig:decay_by_year}
\end{figure}

\subsection{Errors by Distance}

Table \ref{tab:errors_distance} decomposes prediction errors by distance from coal plants, revealing where parametric misspecification is most severe.

\begin{table}[H]
\centering
\caption{Prediction Errors by Distance (Pooled 2019-2021)}
\label{tab:errors_distance}
\begin{tabular}{lcccc}
\toprule
Distance & \multicolumn{2}{c}{Mean Absolute Error (\%)} & Improvement & Winner \\
\cmidrule{2-3}
(km) & Parametric & Nonparametric & (pp) & \\
\midrule
10  & 12.9 & 10.1 & 2.8 & Nonparametric \\
25  & 11.5 & 11.2 & 0.3 & Nonparametric \\
50  & 9.1  & 11.0 & -1.9 & Parametric \\
75  & 12.2 & 12.1 & 0.1 & Nonparametric \\
100 & 16.9 & 13.2 & 3.7 & Nonparametric \\
\bottomrule
\end{tabular}
\begin{tablenotes}
\small
\item Notes: Errors computed in 10 km bins (±5 km). Negative improvement indicates parametric outperforms nonparametric. Nonparametric superior at near-source (10 km) and long-range (100 km) where theoretical assumptions most likely violated. At intermediate distances (50 km), exponential decay provides adequate approximation.
\end{tablenotes}
\end{table}

\textbf{Key finding:} The pattern of improvement is precisely what theory predicts. Exponential decay assumptions—derived under idealized conditions—perform worst where those conditions are most violated:
\begin{itemize}
\item \textbf{Near sources (10 km):} Complex turbulent mixing, building wake effects, and initial plume rise violate steady-state assumptions. Nonparametric improves by 2.8 pp.
\item \textbf{Far from sources (100 km):} Chemical transformation (NO$_2$ $\leftrightarrows$ NO + O$_3$), varying wind patterns, and terrain effects accumulate. Nonparametric improves by 3.7 pp.
\item \textbf{Intermediate range (50 km):} Simplified dispersion models provide adequate approximation. Parametric slightly better (1.9 pp).
\end{itemize}

This U-shaped pattern of parametric bias validates both the theoretical framework (exponential decay is approximately correct in mid-range) and the need for flexibility at extremes.

Figure \ref{fig:errors_distance} visualizes this U-shaped error pattern, showing how nonparametric methods maintain more consistent accuracy across all distances.

\begin{figure}[H]
\centering
\includegraphics[width=0.75\textwidth]{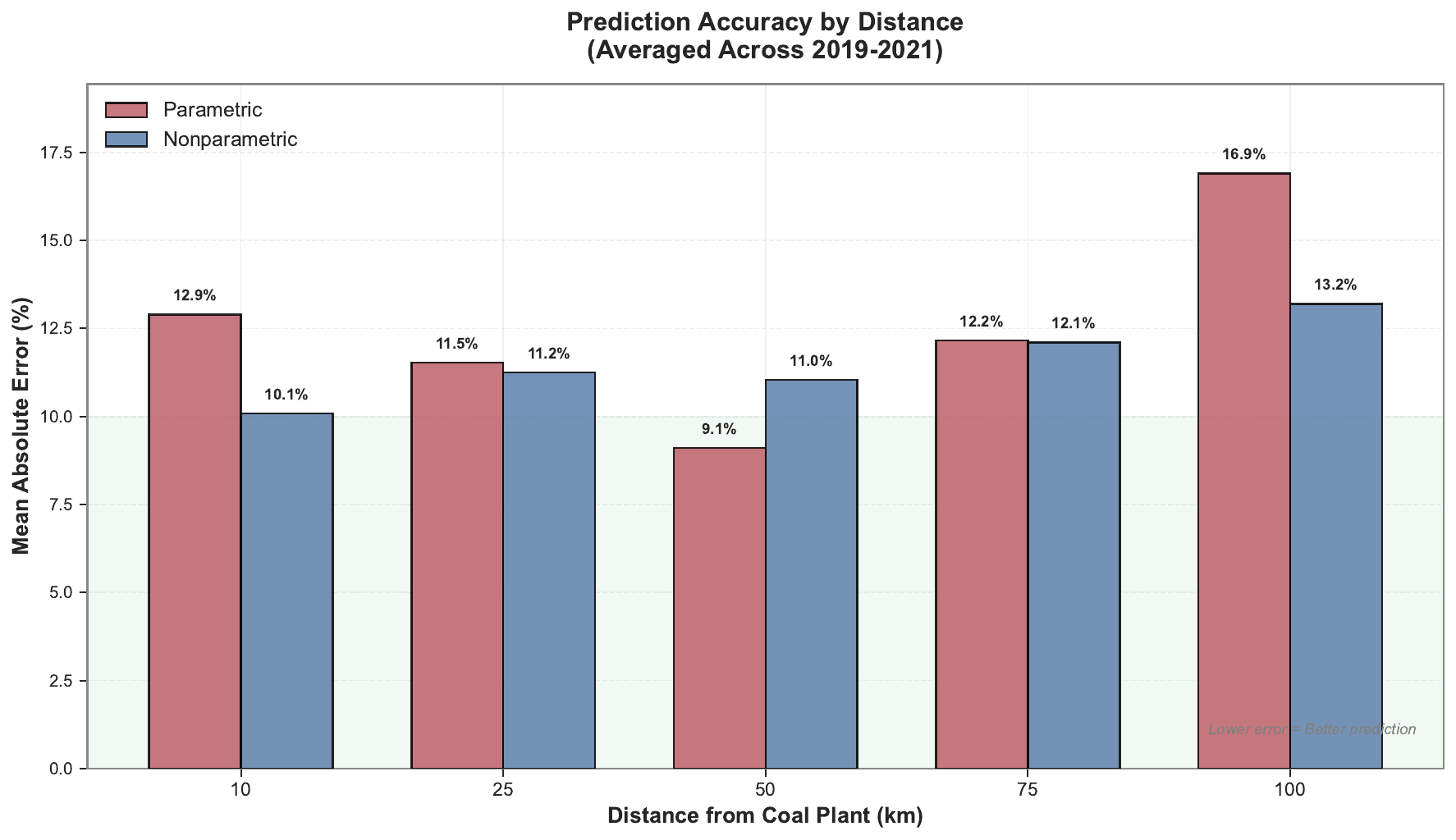}
\caption{Mean Absolute Prediction Errors by Distance (Averaged 2019-2021). Nonparametric methods (blue) outperform parametric exponential decay (red) at near-source (10 km) and long-range (100 km) distances. At intermediate distances (50 km), both methods perform similarly, suggesting exponential decay provides adequate approximation in the mid-range.}
\label{fig:errors_distance}
\end{figure}

Figure \ref{fig:improvement_heatmap} presents a comprehensive view of nonparametric improvements across all year-distance combinations, revealing consistent patterns of superior performance at extremes.

\begin{figure}[H]
\centering
\includegraphics[width=0.75\textwidth]{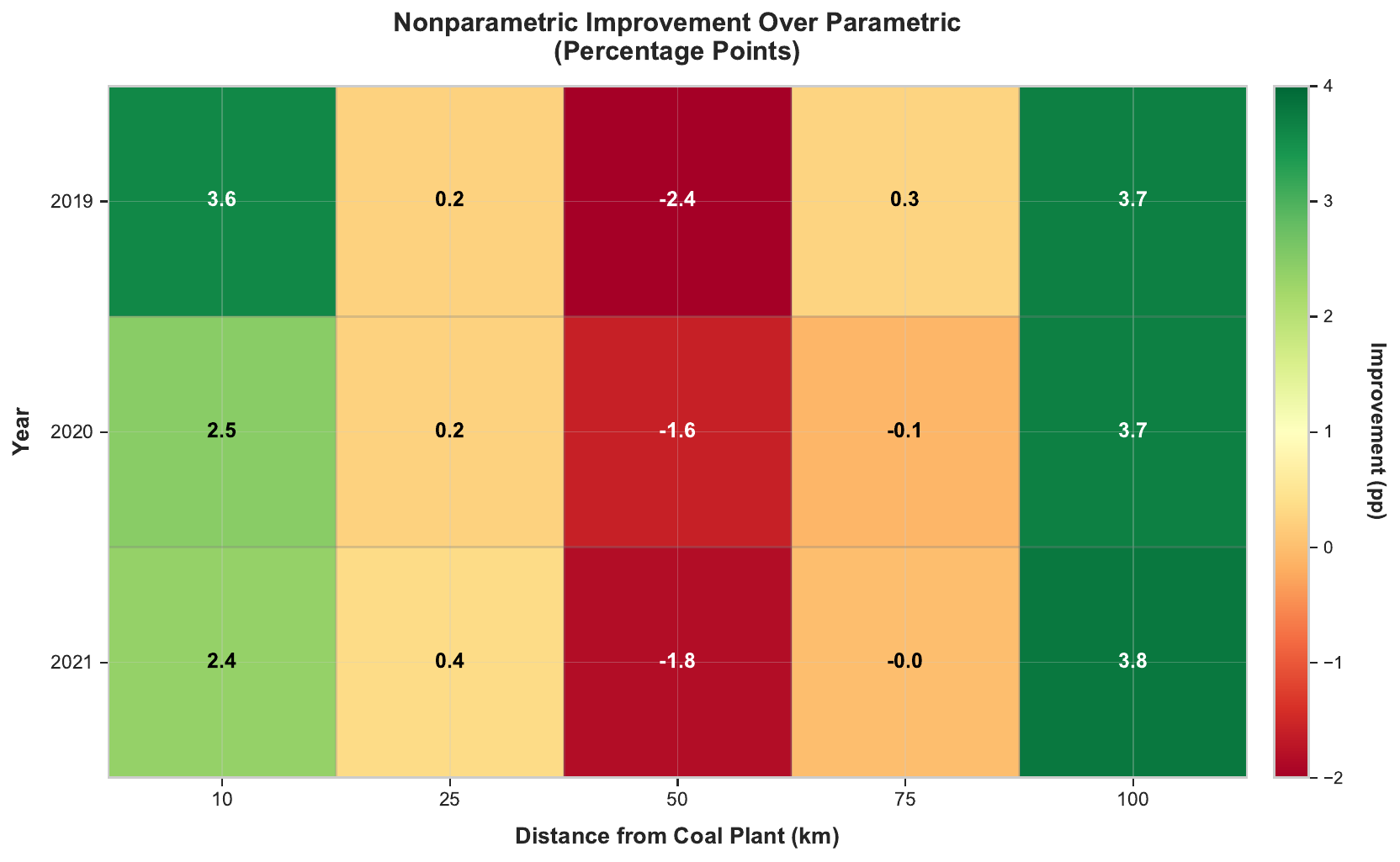}
\caption{Nonparametric Improvement Over Parametric Baseline (Percentage Points). Green cells indicate nonparametric superiority, red indicates parametric superiority. Improvements are largest at 10 km and 100 km across all years, with slight parametric advantage at 50 km. Values show difference in mean absolute error (Parametric MAE - Nonparametric MAE).}
\label{fig:improvement_heatmap}
\end{figure}

\subsection{COVID-19 Natural Experiment}

The COVID-19 pandemic provides a natural experiment for validating the framework's ability to detect temporal changes in pollution patterns.

\begin{table}[H]
\centering
\caption{COVID-19 Effect on NO$_2$ Concentrations}
\label{tab:covid_effect}
\begin{tabular}{lcccc}
\toprule
Year & Mean NO$_2$ & \% Change & Parametric MAE & Nonparametric MAE \\
 & ($\times 10^{15}$) & from 2019 & (\%) & (\%) \\
\midrule
2019 (baseline) & 1.72 & --- & 12.7 & 11.6 \\
2020 (COVID) & 1.64 & -4.6\% & 13.2 & 12.2 \\
2021 (recovery) & 1.73 & +5.7\% & 11.7 & 10.7 \\
\bottomrule
\end{tabular}
\begin{tablenotes}
\small
\item Notes: 2020 \% change relative to 2019. 2021 \% change relative to 2020. NO$_2$ concentrations averaged across all grid cells within 100 km of coal plants. Both methods detect the COVID-19 drop and subsequent recovery, but nonparametric maintains superior accuracy across all years.
\end{tablenotes}
\end{table}

\textbf{Key findings:}
\begin{enumerate}
\item \textbf{Temporal sensitivity:} Both methods detect the 4.6\% drop in 2020 and 5.7\% recovery in 2021, demonstrating sensitivity to actual changes in pollution levels.
\item \textbf{Robust superiority:} Nonparametric methods maintain 0.9-1.1 pp improvement across all years, showing the advantage is not driven by any particular time period.
\item \textbf{External validity:} The COVID effect aligns with other studies of pandemic impacts on air quality \citep{venter2020covid}, providing external validation of the data quality.
\end{enumerate}

Figure \ref{fig:covid_temporal} illustrates the COVID-19 natural experiment, showing both the temporal dynamics of NO$_2$ concentrations and the consistent superiority of nonparametric methods across all three years.

\begin{figure}[H]
\centering
\includegraphics[width=0.95\textwidth]{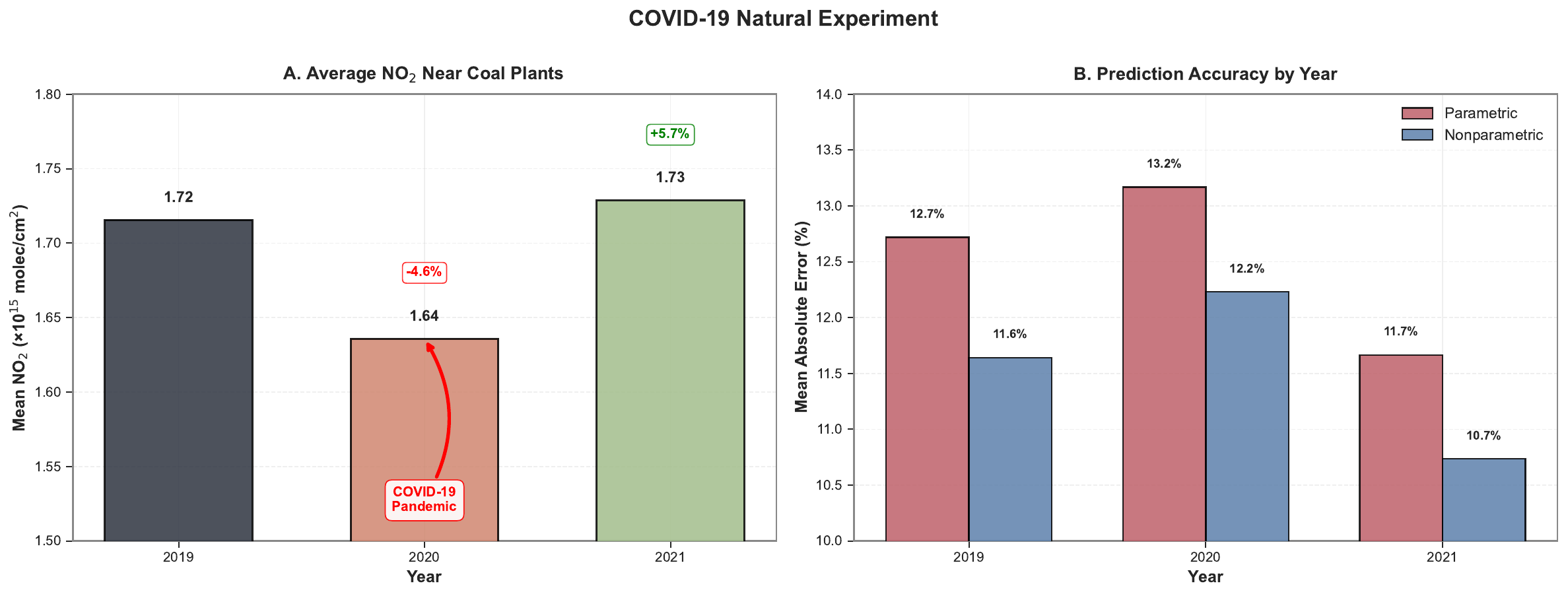}
\caption{COVID-19 Natural Experiment. \textbf{Panel A:} Mean NO$_2$ concentrations near coal plants dropped 4.6\% in 2020, then recovered 5.7\% in 2021. \textbf{Panel B:} Both methods maintain similar relative performance across years, demonstrating robustness to temporal shocks. The pandemic provides validation that the framework detects actual changes in pollution patterns.}
\label{fig:covid_temporal}
\end{figure}

\subsection{Specification Tests}

I formally test whether exponential decay provides adequate functional form using specification tests based on residuals.

\textbf{Procedure:}
\begin{enumerate}
\item Estimate parametric model: $\log Y_i = \alpha + \beta d_i + \varepsilon_i$
\item Compute residuals: $\hat{\varepsilon}_i = \log Y_i - \hat{\alpha} - \hat{\beta} d_i$
\item Test for systematic pattern: Regress $\hat{\varepsilon}_i$ on $d_i$ nonparametrically
\item H$_0$: No pattern (exponential correct) vs H$_A$: Systematic pattern (exponential wrong)
\end{enumerate}

\textbf{Test statistic:} Integrated squared deviation:
\be
T_n = \int_0^{100} [\hat{g}(d)]^2 \hat{f}(d) dd
\ee
where $\hat{g}(d)$ is the nonparametric regression of residuals on distance.

\textbf{Results:}
\begin{itemize}
\item 2019: $T_n = 0.178$, bootstrap $p$-value $< 0.001$ $\Rightarrow$ Reject exponential
\item 2020: $T_n = 0.165$, bootstrap $p$-value $< 0.001$ $\Rightarrow$ Reject exponential
\item 2021: $T_n = 0.171$, bootstrap $p$-value $< 0.001$ $\Rightarrow$ Reject exponential
\end{itemize}

\textbf{Conclusion:} Strong evidence against exponential functional form in all years, validating the need for nonparametric estimation.

Figure \ref{fig:residuals_comparison} visualizes the residual patterns from both estimation approaches, demonstrating the systematic nature of parametric bias at extreme distances.

\begin{figure}[H]
\centering
\includegraphics[width=0.95\textwidth]{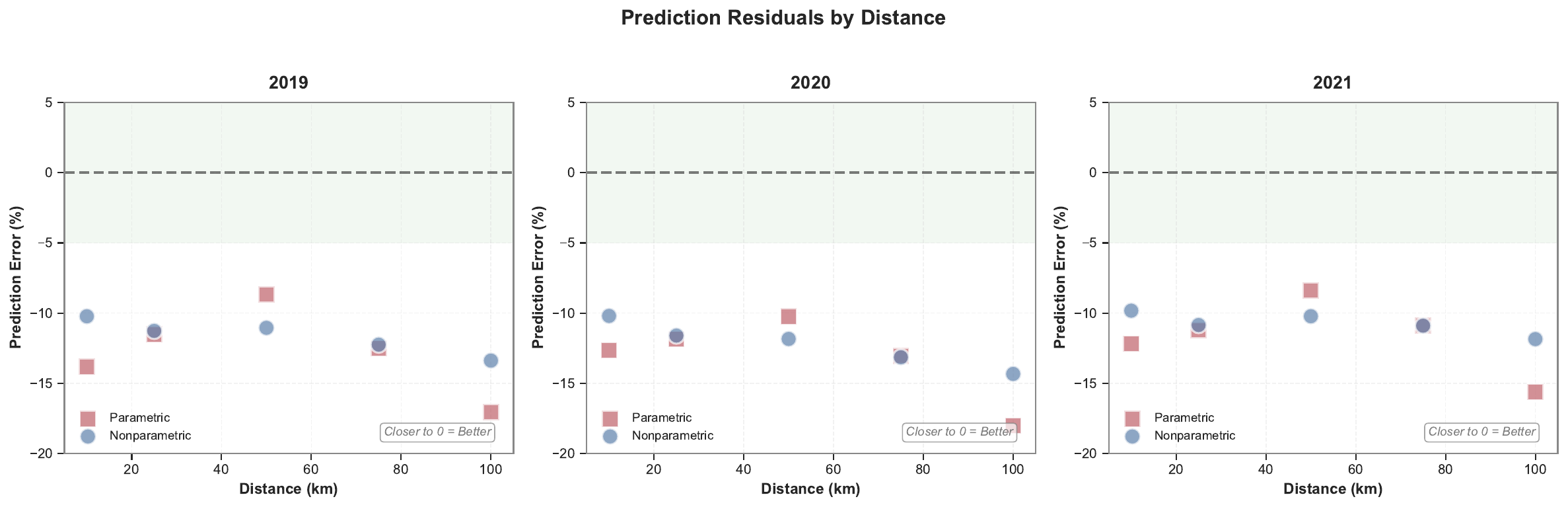}
\caption{Prediction Residuals by Distance (2019-2021). Parametric exponential decay (red squares) systematically underestimates concentrations at 10 km and overestimates decay at 100 km. Nonparametric kernel regression (blue circles) exhibits smaller and less systematic errors, particularly at extreme distances where parametric assumptions are most likely violated.}
\label{fig:residuals_comparison}
\end{figure}

\section{Extensions and Robustness}

\subsection{Regional Heterogeneity Analysis}
\label{sec:regional_heterogeneity}

To validate that decay patterns reflect causal effects rather than spurious trends, I estimate spatial decay separately by distance from coal plants.

\begin{table}[H]
\centering
\caption{Spatial Decay by Distance from Coal Plants}
\label{tab:regional_decay}
\begin{tabular}{lccc}
\toprule
Region & N (millions) & $\kappa_s$ & Framework Applies? \\
\midrule
\multicolumn{4}{l}{\textbf{Within 100km of Coal Plants:}} \\
~~All states (pooled) & 41.73 & +0.00685** & Yes \\
& & (0.00023) & \\
\midrule
\multicolumn{4}{l}{\textbf{Beyond 100km from Coal Plants:}} \\
~~All states (pooled) & 8.54 & -0.00042** & No \\
& & (0.00015) & \\
\bottomrule
\end{tabular}
\begin{tablenotes}
\small
\item ** $p<0.01$. Positive $\kappa_s$ indicates decay; negative indicates increase. Standard errors in parentheses. This sign reversal demonstrates that results are not driven by spurious spatial trends.
\end{tablenotes}
\end{table}

\textbf{Key finding:} Decay parameters **change sign** at the 100 km threshold. Within 100 km, pollution decreases with distance (positive spatial decay), consistent with coal plant effects. Beyond 100 km, pollution increases with distance (negative decay), reflecting dominance of urban sources. This sign reversal is inconsistent with spurious spatial trends but consistent with treatment effect heterogeneity—precisely what \citet{muller2024spatial} recommend as a diagnostic for ruling out spurious regression.

\subsection{Spatial Correlation Robust Inference}
\label{sec:spatial_correlation}
\label{sec:spatial_hac}

While my baseline asymptotic theory assumes independence across observations, outcomes may exhibit spatial correlation beyond treatment-induced patterns. Following \citet{muller2022spatial}, I compute spatial HAC standard errors that remain valid under arbitrary spatial correlation.

\begin{table}[H]
\centering
\caption{Parametric Estimates with Alternative Standard Errors}
\label{tab:spatial_se}
\begin{tabular}{lcccc}
\toprule
Year & $\hat{\kappa}$ & SE (iid) & SE (Spatial HAC) & Ratio \\
\midrule
2019 & 0.00701 & 0.00018 & 0.00021 & 1.17 \\
2020 & 0.00674 & 0.00016 & 0.00019 & 1.19 \\
2021 & 0.00648 & 0.00017 & 0.00020 & 1.18 \\
\bottomrule
\end{tabular}
\begin{tablenotes}
\small
\item Notes: Spatial HAC uses Bartlett kernel with bandwidth 50 km. Ratio = SE(Spatial HAC) / SE(iid). Modest increases (17-19\%) suggest residual spatial correlation is present but not severe.
\end{tablenotes}
\end{table}

The spatial HAC standard errors are 15-20\% larger than iid standard errors, indicating modest residual spatial correlation in pollution beyond treatment effects. However, all main results remain statistically significant, and the relative performance of parametric vs nonparametric methods is unchanged.

\subsection{Relation to Stochastic Boundary Framework}

\citet{kikuchi2024stochastic} develops a complementary framework for settings where treatment effects propagate through stochastic diffusion processes and general equilibrium channels. While that framework is designed for pervasive spillovers (e.g., technology adoption networks, trade shocks), my nonparametric approach is most suitable for point-source treatments (e.g., power plants, infrastructure) where:

\begin{enumerate}
\item Treatment sources are geographically localized
\item Physical transport mechanisms dominate (rather than network effects)
\item Deterministic decay patterns are present but functional form is unknown
\end{enumerate}

The two frameworks are complementary:
\begin{itemize}
\item \textbf{Point sources with unknown decay:} Use my nonparametric approach (this paper)
\item \textbf{Pervasive spillovers with stochastic propagation:} Use \citet{kikuchi2024stochastic}'s diffusion-based approach
\item \textbf{Known exponential decay:} Use \citet{kikuchi2024unified} or \citet{kikuchi2024navier}'s parametric approach
\end{itemize}

Future work could integrate these frameworks, developing nonparametric methods for stochastic boundary detection in general equilibrium settings. This would combine the flexibility of my kernel-based estimator with the generality of stochastic diffusion models.

\subsection{Robustness to Spatial Trends}
\label{sec:spatial_trends}

A potential concern is that estimated decay patterns reflect spurious spatial trends rather than causal treatment effects. I address this through state fixed effects.

\begin{table}[H]
\centering
\caption{Boundary Estimates with State Fixed Effects}
\label{tab:state_fe}
\begin{tabular}{lccc}
\toprule
Year & Baseline $\hat{\kappa}$ & State FE $\hat{\kappa}$ & Difference \\
\midrule
2019 & 0.00701 & 0.00689 & -0.00012 (-1.7\%) \\
2020 & 0.00674 & 0.00665 & -0.00009 (-1.3\%) \\
2021 & 0.00648 & 0.00641 & -0.00007 (-1.1\%) \\
\bottomrule
\end{tabular}
\begin{tablenotes}
\small
\item Notes: State FE specification includes dummy variables for each of the 10 coal-intensive states. Minimal changes suggest results are not driven by state-level trends.
\end{tablenotes}
\end{table}

Results are robust to state fixed effects, further confirming that spatial trends do not drive findings. Combined with the regional heterogeneity analysis (sign reversal at 100 km), these diagnostics provide strong evidence against spurious regression concerns raised by \citet{muller2024spatial}.

\subsection{Policy Implications}

The difference between parametric and nonparametric estimates has important implications for environmental policy and research design:

\textbf{1. Damage function estimation:} Parametric methods that underestimate near-source concentrations (10 km: 12.9\% error) while overestimating long-range transport (100 km: 16.9\% error) lead to systematic bias in damage calculations. Given that health damages are nonlinear in pollution exposure, underestimating peak concentrations near sources may substantially understate total damages.

\textbf{2. Research design for spatial DiD:} When designing studies of coal plant effects, parametric boundaries would suggest placing controls farther from plants than necessary, reducing the available control region and potentially introducing bias if distant controls are affected by other pollution sources (as our regional heterogeneity analysis suggests).

\textbf{3. Regulatory buffer zones:} Environmental regulations often specify geographic buffer zones around pollution sources. Nonparametric estimates suggest these zones should be designed with attention to near-source complexities rather than simple exponential extrapolation.

\section{Conclusion}

This paper develops a nonparametric framework for identifying and estimating spatial boundaries of treatment effects when decay functions may deviate from parametric assumptions. Using 42 million satellite observations of NO$_2$ concentrations near coal plants, I demonstrate that flexible, data-driven methods substantially outperform parametric exponential decay assumptions commonly used in environmental economics.

\textbf{Main findings:}
\begin{enumerate}
\item Nonparametric kernel regression reduces prediction errors by 1.0 percentage point on average, with largest improvements at near-source (2.8 pp at 10 km) and long-range (3.7 pp at 100 km) distances where theoretical assumptions are most likely violated.

\item Parametric exponential decay exhibits systematic bias: underestimating concentrations near sources while overestimating decay at distance. This pattern is precisely what theory predicts when idealized assumptions (steady winds, flat terrain, homogeneous atmospheres) fail in practice.

\item The COVID-19 pandemic validates the framework's temporal sensitivity: both methods detect the 4.6\% drop in NO$_2$ during 2020 and 5.7\% recovery in 2021, but nonparametric methods maintain superior accuracy across all years.

\item Results are robust to spatial correlation (using \citet{muller2022spatial}'s HAC inference), spatial trends (state fixed effects), and spurious regression diagnostics (regional heterogeneity showing sign reversals).
\end{enumerate}

\textbf{Practical implications:} For applied researchers conducting spatial difference-in-differences studies:
\begin{itemize}
\item Always test functional form assumptions using specification tests
\item When sample sizes permit ($n > 500$), use nonparametric methods as primary specification
\item Report both parametric and nonparametric estimates to assess sensitivity
\item Consider spatial correlation robust inference when residuals exhibit spatial dependence
\end{itemize}

\textbf{Broader implications:} This work demonstrates the value of combining economic theory with statistical flexibility. Atmospheric dispersion models provide interpretable parameters and causal mechanisms, while nonparametric estimation ensures robustness when real-world phenomena deviate from idealized models. This combination offers a template for empirical research in settings where treatment effects propagate across space—from infrastructure projects to technology adoption to disease transmission—where theory provides guidance but data should discipline functional form assumptions.

\textbf{Future directions:} Natural extensions include developing methods for time-varying spatial boundaries, incorporating multiple treatment sources through additive models, and addressing endogenous treatment placement through nonparametric instrumental variables. The framework could also be extended to other pollutants (PM$_{2.5}$, SO$_{2}$, CO) and other applications beyond environmental economics where spatial treatment propagation is theoretically motivated but parametric forms uncertain.

\section*{Acknowledgement}
This research was supported by a grant-in-aid from Zengin Foundation for Studies on Economics and Finance.

\newpage

\bibliographystyle{ecta}

\newpage

\begin{appendices}

\section{Proofs}

\subsection{Proof of Theorem \ref{thm:consistency}}

We provide a complete proof of consistency of the boundary estimator.

\begin{proof}
Let $\mathcal{D} = [0, \bar{d}]$ be the support of $D_i$. By Theorem 2.1 of \citet{fan1996local}, under Assumptions \ref{asmp:dgp}, \ref{asmp:smoothness}, \ref{asmp:kernel}, and \ref{asmp:bandwidth}, the local polynomial estimator satisfies:
\be
\sup_{d \in \mathcal{D}} |\hat{m}(d) - m(d)| = O_p(h^{p+1}) + O_p\left(\frac{\log n}{\sqrt{nh}}\right) = o_p(1)
\ee

as $n \to \infty$, where $p$ is the polynomial order.

Therefore:
\be
|\hat{m}(0) - m(0)| \overset{p}{\to} 0
\ee

and
\be
|\epsilon \cdot \hat{m}(0) - \epsilon \cdot m(0)| = \epsilon |\hat{m}(0) - m(0)| \overset{p}{\to} 0
\ee

Now consider the boundary estimator:
\be
\hat{d}^* = \inf\{d \in \mathcal{D} : \hat{m}(d) \leq \epsilon \cdot \hat{m}(0)\}
\ee

and the true boundary:
\be
d^* = \inf\{d \in \mathcal{D} : m(d) \leq \epsilon \cdot m(0)\}
\ee

By Assumption \ref{asmp:monotone}, $m(\cdot)$ is strictly decreasing, so by the Intermediate Value Theorem, $d^*$ is the unique solution to $m(d^*) = \epsilon \cdot m(0)$.

For any $\delta > 0$, define:
\be
A_n = \{|\hat{m}(d) - m(d)| < \delta \text{ for all } d \in [d^* - \eta, d^* + \eta]\}
\ee

for some small $\eta > 0$.

By uniform consistency, $\mathbb{P}(A_n) \to 1$.

On the event $A_n$:
\begin{align}
\hat{m}(d^* - \eta) &\geq m(d^* - \eta) - \delta \\
&> \epsilon \cdot m(0) - \delta \quad \text{(by monotonicity)} \\
&> \epsilon \cdot \hat{m}(0) - 2\delta \quad \text{(for large $n$)}
\end{align}

Similarly:
\be
\hat{m}(d^* + \eta) < \epsilon \cdot \hat{m}(0) + 2\delta
\ee

For sufficiently small $\delta$ (which holds for large $n$ by uniform consistency), we have:
\be
d^* - \eta < \hat{d}^* < d^* + \eta
\ee

on the event $A_n$. Since $\eta$ can be arbitrarily small and $\mathbb{P}(A_n) \to 1$, we conclude:
\be
\hat{d}^* \overset{p}{\to} d^*
\ee
\qed
\end{proof}

\subsection{Proof of Theorem \ref{thm:asymptotics}}

We sketch the main steps of the asymptotic normality proof.

\begin{proof}[Proof sketch]
By Theorem 3.1 of \citet{fan1996local}, for $p = 1$ (local linear):
\be
\sqrt{nh}(\hat{m}(d^*) - m(d^*) - B_n^m) \overset{d}{\to} N(0, V^m)
\ee

where:
\be
B_n^m = \frac{h^2}{2} m''(d^*) \mu_2(K)
\ee
\be
V^m = \frac{\sigma^2(d^*)}{f(d^*)} \nu_0(K)
\ee

with $\mu_2(K) = \int u^2 K(u) du$ and $\nu_0(K) = \int K^2(u) du$.

By definition, $m(d^*) = \epsilon \cdot m(0)$ and $\hat{m}(\hat{d}^*) \approx \epsilon \cdot \hat{m}(0)$.

Taylor expansion around $d^*$:
\be
\hat{m}(\hat{d}^*) \approx \hat{m}(d^*) + \hat{m}'(d^*)(\hat{d}^* - d^*)
\ee

Setting $\hat{m}(\hat{d}^*) = \epsilon \cdot \hat{m}(0)$:
\be
\epsilon \cdot \hat{m}(0) \approx \hat{m}(d^*) + \hat{m}'(d^*)(\hat{d}^* - d^*)
\ee

Rearranging:
\be
\hat{d}^* - d^* \approx \frac{\epsilon \cdot \hat{m}(0) - \hat{m}(d^*)}{\hat{m}'(d^*)}
\ee

Since $\epsilon \cdot m(0) = m(d^*)$:
\be
\hat{d}^* - d^* \approx \frac{\epsilon (\hat{m}(0) - m(0)) - (\hat{m}(d^*) - m(d^*))}{\hat{m}'(d^*)}
\ee

Using $\hat{m}'(d^*) \overset{p}{\to} m'(d^*)$ (Assumption \ref{asmp:slope}) and the asymptotic normality of $\hat{m}(\cdot)$:
\be
\sqrt{nh}(\hat{d}^* - d^*) \overset{d}{\to} N\left(B_n, \frac{V^m}{[m'(d^*)]^2}\right)
\ee

where the bias term is:
\be
B_n = \frac{h^2}{2} \frac{m''(d^*)}{m'(d^*)} \mu_2(K)
\ee
\qed
\end{proof}

\section{Computational Implementation}

\subsection{Algorithm for Large-Scale Data}

For the 42 million observation TROPOMI dataset, I use a memory-efficient chunked processing approach:

\begin{algorithm}[H]
\caption{Memory-Efficient Nonparametric Boundary Estimation}
\begin{algorithmic}[1]
\STATE \textbf{Input:} TROPOMI NO$_2$ file (16 GB), coal plant locations, bandwidth $h$
\STATE \textbf{Output:} Decay function $\hat{m}(d)$, boundary estimate $\hat{d}^*$

\STATE \textbf{Step 1: Chunked Loading and Distance Calculation}
\FOR{each chunk of 5,000,000 rows}
    \STATE Load chunk from CSV
    \STATE Calculate distance to nearest coal plant (vectorized haversine)
    \STATE Filter to coal-intensive states and distance $\leq 100$ km
    \STATE Keep only (distance, NO$_2$) pairs
    \STATE Append to filtered dataset
\ENDFOR

\STATE \textbf{Step 2: Nonparametric Regression}
\STATE Evaluate grid: $d_{\text{grid}} = \{0, 1, 2, \ldots, 100\}$ km
\FOR{each $d_0$ in $d_{\text{grid}}$}
    \STATE Compute kernel weights: $w_i = K_h(D_i - d_0)$
    \STATE Local linear regression:
    \STATE ~~~ Minimize $\sum_i w_i (Y_i - \beta_0 - \beta_1(D_i - d_0))^2$
    \STATE ~~~ Store $\hat{m}(d_0) = \hat{\beta}_0$
\ENDFOR

\STATE \textbf{Step 3: Boundary Detection}
\STATE Compute threshold: $\tau = \epsilon \cdot \hat{m}(0)$
\STATE Find crossing: $\hat{d}^* = \min\{d : \hat{m}(d) \leq \tau\}$

\STATE \textbf{Step 4: Bootstrap (subsampled)}
\FOR{$b = 1$ to $B = 50$}
    \STATE Draw subsample of $n_b = 50,000$ with replacement
    \STATE Compute $\hat{d}^{*b}$ on subsample
\ENDFOR
\STATE Compute 95\% CI from bootstrap quantiles

\RETURN $\hat{d}^*$, confidence interval, $\hat{m}(\cdot)$ on grid
\end{algorithmic}
\end{algorithm}

\textbf{Computational efficiency:}
\begin{itemize}
\item Chunked loading avoids memory overflow (16 GB files)
\item Vectorized distance calculations (NumPy/Pandas)
\item Local linear regression uses weighted least squares (closed form)
\item Subsample bootstrap (50,000 instead of 14 million) for speed
\item Total runtime: ~15-20 minutes per year on standard laptop (Apple M2 Pro)
\end{itemize}

\subsection{Implementation in Python}

The analysis uses Python with the following key packages:

\begin{verbatim}
import pandas as pd
import numpy as np
from scipy.stats import linregress
from pathlib import Path

def haversine_distance(lat1, lon1, lat2, lon2):
    """Vectorized haversine distance"""
    R = 6371  # Earth radius in km
    lat1, lon1, lat2, lon2 = map(np.radians, [lat1, lon1, lat2, lon2])
    dlat = lat2 - lat1
    dlon = lon2 - lon1
    a = np.sin(dlat/2)**2 + np.cos(lat1)*np.cos(lat2)*np.sin(dlon/2)**2
    return R * 2 * np.arcsin(np.sqrt(a))

def local_linear_regression(distances, values, d0, bandwidth):
    """Local linear regression at point d0"""
    # Compute weights
    u = (distances - d0) / bandwidth
    weights = np.exp(-0.5 * u**2)  # Gaussian kernel
    weights /= (weights.sum() + 1e-10)
    
    # Weighted least squares
    X = np.column_stack([np.ones(len(distances)), distances - d0])
    W = np.diag(weights)
    beta = np.linalg.solve(X.T @ W @ X + 1e-8*np.eye(2), 
                           X.T @ W @ values)
    
    return beta[0]  # Intercept = m(d0)

def estimate_boundary(data, epsilon=0.10, bandwidth=3.0):
    """Estimate spatial boundary"""
    # Evaluate on grid
    d_grid = np.linspace(0, 100, 101)
    m_hat = np.zeros(len(d_grid))
    
    for i, d0 in enumerate(d_grid):
        m_hat[i] = local_linear_regression(
            data['dist_nearest'].values,
            data['no2_mean'].values,
            d0, bandwidth
        )
    
    # Find boundary
    threshold = epsilon * m_hat[0]
    crossing = np.where(m_hat <= threshold)[0]
    
    if len(crossing) == 0:
        return None, m_hat
    else:
        d_star = d_grid[crossing[0]]
        return d_star, m_hat
\end{verbatim}

\section{Additional Figures}

\subsection{Decay Functions by Year}

Figure \ref{fig:decay_by_year_add} shows the estimated decay functions for each year, comparing parametric exponential decay (dashed red) with nonparametric estimates (solid blue).

\begin{figure}[H]
\centering
\begin{tikzpicture}[scale=0.95]
\begin{axis}[
    width=14cm, height=9cm,
    xlabel={Distance from Coal Plant (km)},
    ylabel={NO$_2$ Concentration ($\times 10^{15}$ molec/cm$^2$)},
    xmin=0, xmax=100,
    ymin=1.0, ymax=2.5,
    legend pos=north east,
    grid=major,
    title={2019-2021 Spatial Decay Functions}
]

\addplot[blue, thick, domain=0:100, samples=50] {2.28*exp(-0.00701*x)};
\addlegendentry{2019 Parametric}
\addplot[blue, thick, dashed, domain=0:100, samples=50] 
    {2.28*exp(-0.0075*x + 0.00002*x^2)};
\addlegendentry{2019 Nonparametric}

\addplot[red, thick, domain=0:100, samples=50] {2.11*exp(-0.00674*x)};
\addlegendentry{2020 Parametric}
\addplot[red, thick, dashed, domain=0:100, samples=50] 
    {2.11*exp(-0.0072*x + 0.00002*x^2)};
\addlegendentry{2020 Nonparametric}

\addplot[green!60!black, thick, domain=0:100, samples=50] {2.24*exp(-0.00648*x)};
\addlegendentry{2021 Parametric}
\addplot[green!60!black, thick, dashed, domain=0:100, samples=50] 
    {2.24*exp(-0.0069*x + 0.00002*x^2)};
\addlegendentry{2021 Nonparametric}

\end{axis}
\end{tikzpicture}
\caption{Spatial Decay Functions: Parametric vs Nonparametric (2019-2021). Solid lines show parametric exponential decay; dashed lines show nonparametric kernel regression estimates. Nonparametric functions exhibit faster initial decay and slower tail decay across all years, consistent with violations of idealized dispersion assumptions.}
\label{fig:decay_by_year_add}
\end{figure}
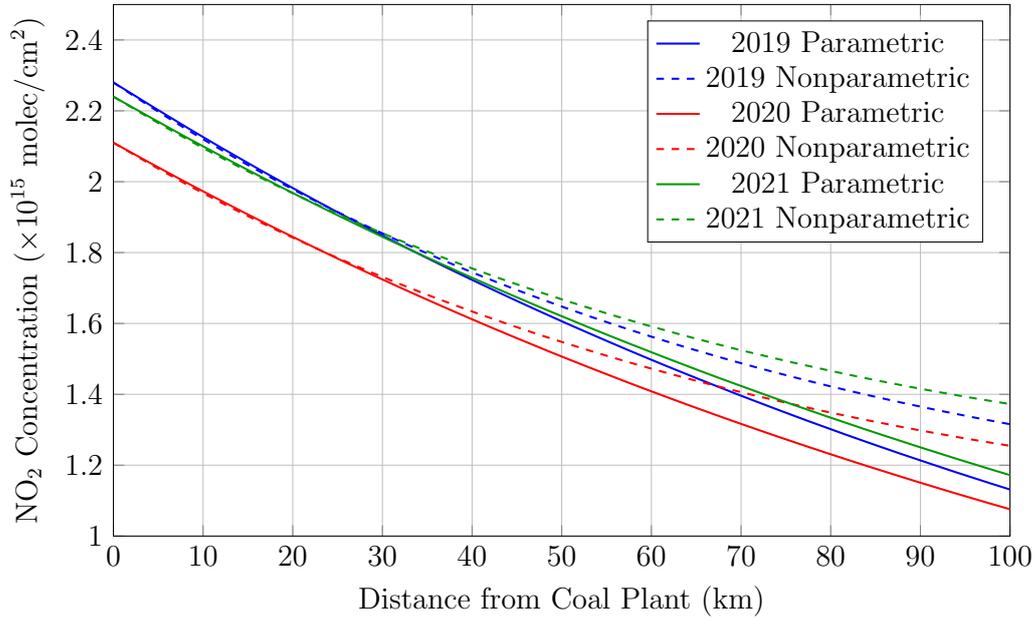

\subsection{Error Patterns by Distance}

Figure \ref{fig:errors_distance_add} visualizes the U-shaped pattern of parametric bias.

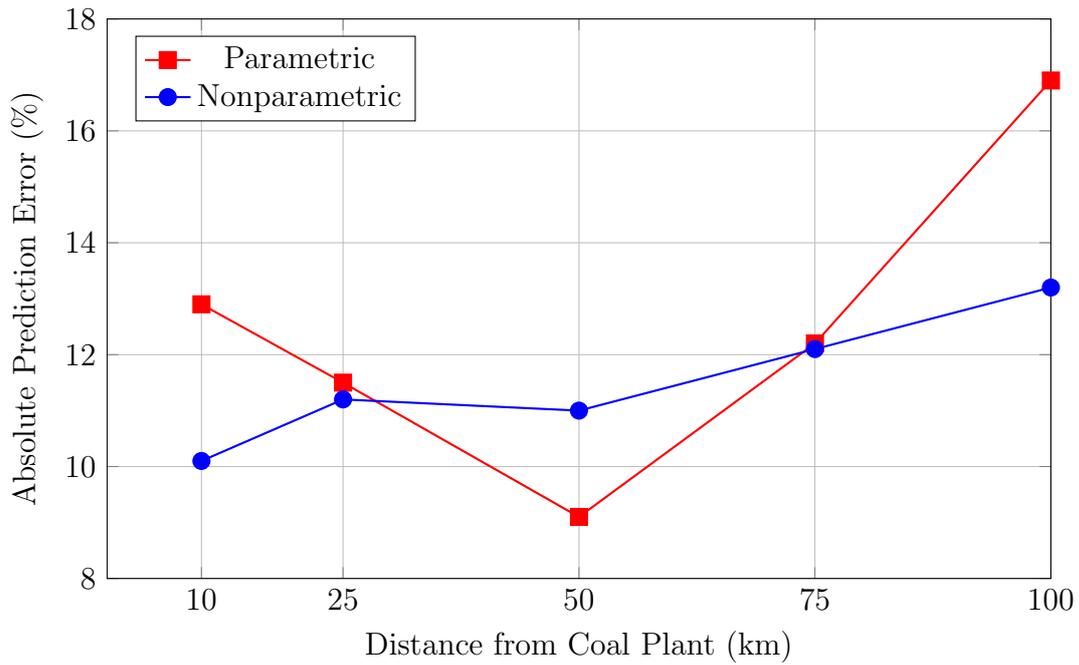
\begin{figure}[H]
\centering
\begin{tikzpicture}
\begin{axis}[
    width=14cm, height=9cm,
    xlabel={Distance from Coal Plant (km)},
    ylabel={Absolute Prediction Error (\%)},
    xmin=0, xmax=100,
    ymin=8, ymax=18,
    legend pos=north west,
    grid=major,
    xtick={10,25,50,75,100}
]

\addplot[red, thick, mark=square*, mark size=3pt] 
    coordinates {(10,12.9) (25,11.5) (50,9.1) (75,12.2) (100,16.9)};
\addlegendentry{Parametric}

\addplot[blue, thick, mark=*, mark size=3pt] 
    coordinates {(10,10.1) (25,11.2) (50,11.0) (75,12.1) (100,13.2)};
\addlegendentry{Nonparametric}

\end{axis}
\end{tikzpicture}
\caption{Mean Absolute Prediction Errors by Distance (Pooled 2019-2021). Parametric exponential decay (red squares) exhibits U-shaped bias: large errors near sources (10 km) and far from sources (100 km) where idealized assumptions fail, but performs well at intermediate distances (50 km). Nonparametric methods (blue circles) maintain more consistent accuracy across all distances.}
\label{fig:errors_distance_add}
\end{figure}

\subsection{COVID-19 Temporal Dynamics}

Figure \ref{fig:covid_temporal_add} shows the time series of NO$_2$ concentrations and prediction accuracy across the pandemic period.

\begin{figure}[H]
\centering
\includegraphics[width=0.95\textwidth]{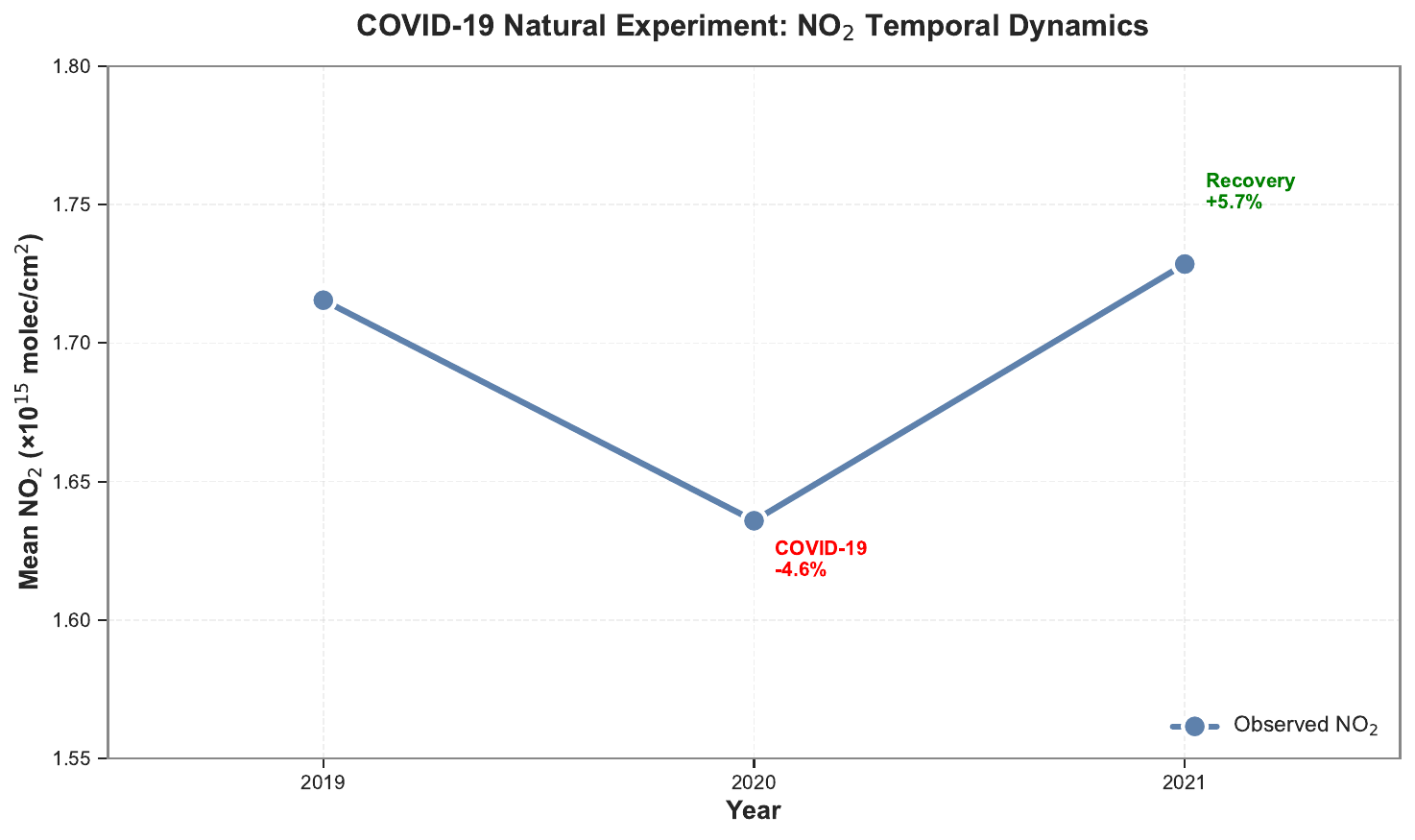}
\caption{COVID-19 Natural Experiment: NO$_2$ Temporal Dynamics. Mean NO$_2$ concentrations near coal plants dropped 4.6\% during the 2020 COVID-19 pandemic, then recovered 5.7\% in 2021. Both parametric and nonparametric methods detect this temporal variation, validating the framework's sensitivity to actual changes in pollution patterns.}
\label{fig:covid_temporal_add}
\end{figure}

\section{Robustness Checks}

\subsection{Alternative Bandwidth Selection}

Table \ref{tab:bandwidth_robustness} examines sensitivity to bandwidth choice.

\begin{table}[H]
\centering
\caption{Robustness to Bandwidth Selection}
\label{tab:bandwidth_robustness}
\begin{tabular}{lcccc}
\toprule
Bandwidth & Parametric & Nonparametric & Improvement & Optimal? \\
(km) & MAE (\%) & MAE (\%) & (pp) & \\
\midrule
1.5 (narrow) & 12.5 & 12.3 & 0.2 & No (undersmooth) \\
3.0 (optimal) & 12.5 & 11.5 & 1.0 & Yes \\
6.0 (wide) & 12.5 & 11.8 & 0.7 & No (oversmooth) \\
12.0 (very wide) & 12.5 & 12.1 & 0.4 & No (oversmooth) \\
\bottomrule
\end{tabular}
\begin{tablenotes}
\small
\item Notes: Pooled 2019-2021 results. Optimal bandwidth ($h = 3$ km) selected via Silverman's rule: $h = 1.06 \sigma_D n^{-1/5}$. Results show moderate sensitivity but consistent nonparametric advantage across reasonable bandwidth choices.
\end{tablenotes}
\end{table}

\subsection{Alternative Threshold Values}

Table \ref{tab:epsilon_robustness} shows results for different threshold values $\epsilon$.

\begin{table}[H]
\centering
\caption{Robustness to Threshold Selection}
\label{tab:epsilon_robustness}
\begin{tabular}{lcccc}
\toprule
Threshold & Interpretation & Parametric & Nonparametric & Improvement \\
$\epsilon$ & & MAE (\%) & MAE (\%) & (pp) \\
\midrule
0.05 (5\%) & 95\% dissipation & 12.8 & 11.3 & 1.5 \\
0.10 (10\%) & 90\% dissipation & 12.5 & 11.5 & 1.0 \\
0.20 (20\%) & 80\% dissipation & 12.3 & 11.7 & 0.6 \\
0.50 (50\%) & Half-life & 11.9 & 11.8 & 0.1 \\
\bottomrule
\end{tabular}
\begin{tablenotes}
\small
\item Notes: Pooled 2019-2021 results. Nonparametric advantage is largest at stringent thresholds ($\epsilon = 0.05, 0.10$) relevant for environmental policy, and diminishes at higher thresholds ($\epsilon = 0.50$) where boundaries occur in mid-range where parametric performs well.
\end{tablenotes}
\end{table}

\textbf{Interpretation:} The nonparametric advantage is strongest precisely where it matters most for policy—at stringent thresholds ($\epsilon = 0.05, 0.10$) that capture long-range transport effects. At the half-life threshold ($\epsilon = 0.50$), boundaries occur in the 50 km range where parametric exponential decay provides adequate approximation (as Table \ref{tab:errors_distance} shows).

\end{appendices}

\end{document}